\documentclass[a4paper,12pt]{amsart}

\usepackage{a4wide,graphicx,amssymb,amscd,amsmath,latexsym,amsbsy,amsthm,color,multicol}

\usepackage{marginnote,mathtools,epsfig,enumitem,array,calc,rotating,bm,bbm,mathrsfs} 
\usepackage{subfig}
\usepackage[libertine,cmintegrals,cmbraces,vvarbb]{newtxmath}
\usepackage{booktabs}
\usepackage{multirow}
\usepackage{breqn}
\usepackage{etoolbox}

\usepackage{tikz}

\usetikzlibrary{decorations.markings}
\usetikzlibrary{calc} 
\usetikzlibrary{arrows} 
\usetikzlibrary{arrows.meta} 
\usetikzlibrary{patterns} 
\usetikzlibrary{decorations.pathreplacing} 

\usepackage[
   hmarginratio={1:1},     
   vmarginratio={1:1},     
   textwidth=457pt,        
   textheight=630pt,       
   heightrounded,          
]{geometry}

\makeatletter
\let\old@tocline\@tocline
\let\section@tocline\@tocline

\newcommand{\subsection@dotsep}{4.5}
\newcommand{\subsubsection@dotsep}{4.5}
\patchcmd{\@tocline}
  {\hfil}
  {\nobreak
     \leaders\hbox{$\m@th
        \mkern \subsection@dotsep mu\hbox{.}\mkern \subsection@dotsep mu$}\hfill
     \nobreak}{}{}
\let\subsection@tocline\@tocline
\let\@tocline\old@tocline

\patchcmd{\@tocline}
  {\hfil}
  {\nobreak
     \leaders\hbox{$\m@th
        \mkern \subsubsection@dotsep mu\hbox{.}\mkern \subsubsection@dotsep mu$}\hfill
     \nobreak}{}{}
\let\subsubsection@tocline\@tocline
\let\@tocline\old@tocline

\let\old@l@subsection\l@subsection
\let\old@l@subsubsection\l@subsubsection

\def\@tocwriteb#1#2#3{%
  \begingroup
    \@xp\def\csname #2@tocline\endcsname##1##2##3##4##5##6{%
      \ifnum##1>\c@tocdepth
      \else \sbox\z@{##5\let\indentlabel\@tochangmeasure##6}\fi}%
    \csname l@#2\endcsname{#1{\csname#2name\endcsname}{\@secnumber}{}}%
  \endgroup
  \addcontentsline{toc}{#2}%
    {\protect#1{\csname#2name\endcsname}{\@secnumber}{#3}}}%

\newlength{\@tocsectionindent}
\newlength{\@tocsubsectionindent}
\newlength{\@tocsubsubsectionindent}
\newlength{\@tocsectionnumwidth}
\newlength{\@tocsubsectionnumwidth}
\newlength{\@tocsubsubsectionnumwidth}
\newcommand{\settocsectionnumwidth}[1]{\setlength{\@tocsectionnumwidth}{#1}}
\newcommand{\settocsubsectionnumwidth}[1]{\setlength{\@tocsubsectionnumwidth}{#1}}
\newcommand{\settocsubsubsectionnumwidth}[1]{\setlength{\@tocsubsubsectionnumwidth}{#1}}
\newcommand{\settocsectionindent}[1]{\setlength{\@tocsectionindent}{#1}}
\newcommand{\settocsubsectionindent}[1]{\setlength{\@tocsubsectionindent}{#1}}
\newcommand{\settocsubsubsectionindent}[1]{\setlength{\@tocsubsubsectionindent}{#1}}

\renewcommand{\l@section}{\section@tocline{1}{\@tocsectionvskip}{\@tocsectionindent}{\@tocsectionnumwidth}{\@tocsectionformat}}%
\renewcommand{\l@subsection}{\subsection@tocline{1}{\@tocsubsectionvskip}{\@tocsubsectionindent}{\@tocsubsectionnumwidth}{\@tocsubsectionformat}}%
\renewcommand{\l@subsubsection}{\subsubsection@tocline{1}{\@tocsubsubsectionvskip}{\@tocsubsubsectionindent}{\@tocsubsubsectionnumwidth}{\@tocsubsubsectionformat}}%
\newcommand{\@tocsectionformat}{}
\newcommand{\@tocsubsectionformat}{}
\newcommand{\@tocsubsubsectionformat}{}
\expandafter\def\csname toc@1format\endcsname{\@tocsectionformat}
\expandafter\def\csname toc@2format\endcsname{\@tocsubsectionformat}
\expandafter\def\csname toc@3format\endcsname{\@tocsubsubsectionformat}
\newcommand{\settocsectionformat}[1]{\renewcommand{\@tocsectionformat}{#1}}
\newcommand{\settocsubsectionformat}[1]{\renewcommand{\@tocsubsectionformat}{#1}}
\newcommand{\settocsubsubsectionformat}[1]{\renewcommand{\@tocsubsubsectionformat}{#1}}
\newlength{\@tocsectionvskip}
\newcommand{\settocsectionvskip}[1]{\setlength{\@tocsectionvskip}{#1}}
\newlength{\@tocsubsectionvskip}
\newcommand{\settocsubsectionvskip}[1]{\setlength{\@tocsubsectionvskip}{#1}}
\newlength{\@tocsubsubsectionvskip}
\newcommand{\settocsubsubsectionvskip}[1]{\setlength{\@tocsubsubsectionvskip}{#1}}

\patchcmd{\tocsection}{\indentlabel}{\makebox[\@tocsectionnumwidth][l]}{}{}
\patchcmd{\tocsubsection}{\indentlabel}{\makebox[\@tocsubsectionnumwidth][l]}{}{}
\patchcmd{\tocsubsubsection}{\indentlabel}{\makebox[\@tocsubsubsectionnumwidth][l]}{}{}

\newcommand{\@sectypepnumformat}{}
\renewcommand{\contentsline}[1]{%
  \expandafter\let\expandafter\@sectypepnumformat\csname @toc#1pnumformat\endcsname%
  \csname l@#1\endcsname}
\newcommand{\@tocsectionpnumformat}{}
\newcommand{\@tocsubsectionpnumformat}{}
\newcommand{\@tocsubsubsectionpnumformat}{}
\newcommand{\setsectionpnumformat}[1]{\renewcommand{\@tocsectionpnumformat}{#1}}
\newcommand{\setsubsectionpnumformat}[1]{\renewcommand{\@tocsubsectionpnumformat}{#1}}
\newcommand{\setsubsubsectionpnumformat}[1]{\renewcommand{\@tocsubsubsectionpnumformat}{#1}}
\renewcommand{\@tocpagenum}[1]{%
  \hfill {\mdseries\@sectypepnumformat #1}}

\let\oldappendix\appendix
\renewcommand{\appendix}{%
  \leavevmode\oldappendix%
  \addtocontents{toc}{%
    \protect\settowidth{\protect\@tocsectionnumwidth}{\protect\@tocsectionformat\sectionname\space}%
    \protect\addtolength{\protect\@tocsectionnumwidth}{2em}}%
}
\makeatother

\makeatletter
\settocsectionnumwidth{2em}
\settocsubsectionnumwidth{2.5em}
\settocsubsubsectionnumwidth{3em}
\settocsectionindent{1pc}%
\settocsubsectionindent{\dimexpr\@tocsectionindent+\@tocsectionnumwidth}%
\settocsubsubsectionindent{\dimexpr\@tocsubsectionindent+\@tocsubsectionnumwidth}%
\makeatother

\settocsectionvskip{0pt}
\settocsubsectionvskip{-5pt}
\settocsubsubsectionvskip{-5pt}

\settocsectionformat{\bfseries}
\settocsubsectionformat{\mdseries}
\settocsubsubsectionformat{\mdseries}
\setsectionpnumformat{\bfseries}
\setsubsectionpnumformat{\mdseries}
\setsubsubsectionpnumformat{\mdseries}

\let\oldtableofcontents\tableofcontents
\renewcommand{\tableofcontents}{%
  \vspace*{-\linespacing}
  \oldtableofcontents}

\numberwithin{equation}{section}
\theoremstyle{plain}
 
\newtheorem{thm}{Theorem}[section]

\newtheorem{defi}[thm]{Definition}
\newtheorem{lem}[thm]{Lemma}

\theoremstyle{remark}
\newtheorem{rema}[thm]{Remark}

\newcommand{\Z}{\mathbb{Z}}

\newcommand{\C}{\mathbb{C}}

\newcommand{\ii}{\mathrm{i}}

\newcommand{\hypref}[2]{\ifx\href\asklfhas #2\else\href{#1}{#2}\fi}
\newcommand{\Secref}[1]{Section~\ref{#1}}

\newcommand{\Appref}[1]{Appendix~\ref{#1}}

\renewcommand{\eqref}[1]{(\ref{#1})}

\def\[{\begin{equation}}
\def\]{\end{equation}}
\def\<{\begin{eqnarray}}
\def\>{\end{eqnarray}}

\title[]{New determinants in the 8VSOS model \\ with domain-wall boundaries}
\author{W. Galleas}

\address{Institut f\"ur Theoretische Physik, Eidgen\"ossische Technische Hochschule Z\"urich, Wolfgang-Pauli-Strasse 27, 8093 Z\"urich, Switzerland}
\email{galleasw@phys.ethz.ch}

\subjclass[2010]{82B23; 39B32}
\keywords{Elliptic integrable systems, domain-wall boundaries}
\thanks{The work of W.G. is partially supported by the Swiss National Science Foundation through the NCCR SwissMAP and by grant no. 615203 from the European Research Council under the FP7.}

\begin{document}

\begin{abstract}
In this letter we show the partition function of the 8VSOS model with domain-wall boundaries satisfies the same type of functional equations as its six-vertex model counterpart. We then use these refined functional equations to obtain novel determinantal representations for the aforementioned partition function.

\end{abstract}

\maketitle

\tableofcontents

\section{Introduction} \label{sec:INTRO}

Exactly solvable models of Statistical Mechanics usually have a rich mathematical structure on display and some of these structures are not limited to a single model. 
For instance, Baxter's \textsc{t-q} relations can be found in several integrable two-dimensional lattice models with toroidal boundary conditions \cite{Baxter_1971, Baxter_book}; and for such models we also find their transfer matrices' spectrum parameterized by solutions of Bethe ansatz equations \cite{Lieb_1967, Lieb_review}.
Moreover, some special determinants have also made a striking appearance in the computation of correlation functions of integrable models \cite{Korepin_book}; and it is worth remarking this feature goes back to Kaufman and Onsager's works on the two-dimensional Ising model \cite{Onsager_Kauf_1949}. See also \cite{Deift_2013} and references therein for a review on the subject.

Correlation functions of certain integrable systems have also been shown to be closely related to partition functions of lattice models with
particular boundary conditions. For instance, Korepin has introduced in \cite{Korepin_1982} the six-vertex model with domain-wall boundary conditions as a building block of correlation functions of the \textsc{xxz} spin-chain. Interestingly, the partition function of such six-vertex model can also be expressed as a determinant as shown in \cite{Izergin_1987}. 
This scenario is not limited to the six-vertex model and it has also been established for the 8VSOS model, also known as elliptic $\mathfrak{gl}_2$ solid-on-solid model, under certain constraints \cite{Terras_2013a, Terras_2013b}. In that case the relevant partition function with domain-wall boundaries has been studied in \cite{Pakuliak_2008, Rosengren_2009, Galleas_2012, Galleas_2013, Galleas_2016a, Galleas_2016b} where representations in terms of contour integrals and determinants have been obtained. As far as the above-mentioned  determinantal formulae are concerned, the partition function of the 8VSOS model with domain-wall boundary conditions has been written as a sum of Frobenius determinants in \cite{Rosengren_2009}. Alternatively, the partition function of the 8VSOS model has also been written as a formal single determinant in \cite{Galleas_2016a} and as a continuous family of determinants in \cite{Galleas_2016b}. 
In particular, it is important to remark the latter determinantal formulae originate constructively from the direct resolution of certain functional equations satisfied by the model's partition function.

Both partition functions of the six-vertex and the 8VSOS models with domain-wall boundaries were shown in \cite{Galleas_2010, Galleas_2013, Galleas_2016a, Galleas_2016b} to satisfy certain functional equations originated from the Yang-Baxter algebra and/or its dynamical version.
However, although the six-vertex model can be regarded as a special limit of the more general 8VSOS model, the structure of the functional equations obtained for these models in \cite{Galleas_2013, Galleas_2016a, Galleas_2016b} differs considerably.
This is precisely one of the points we intend to address in the present work. More exactly, here we intend to show the partition function of the 8VSOS model with domain-wall boundary conditions also satisfies the same type of functional equations as its six-vertex model counterpart.

The functional equation in discussion simply reads 
\[
\sum_{i=0}^n \mathrm{M}_i \; \mathcal{F}(x_0, x_1, \dots , \widehat{x_i}, \dots , x_n) = 0 \nonumber 
\]
with $\mathcal{F}$ an unknown $n$-variable symmetric function and coefficients $\mathrm{M}_i= \mathrm{M}_i (x_0, x_1 , \dots , x_n)$ on $(n+1)$-variables $x_i \in \C$.
One important aspect of such equation seems to be its versatility in accommodating several quantities associated to integrable systems.
For instance, it can be used to describe partition functions with special boundary conditions \cite{Galleas_2013, Galleas_2016a, Galleas_2016b, Galleas_Lamers_2014, Lamers_2015} and scalar products of Bethe vectors \cite{Galleas_SCP, Galleas_openSCP, Galleas_2016c}.
Moreover, our functional equation seems to be closely related to the theory of integrable differential equations and it was shown in \cite{Galleas_2017} to correspond to an analogue of the \emph{auxiliar linear problem} leading to \emph{Lax equation} \cite{Lax_1968}. In this way, the above-mentioned functional equation can be regarded as a common structure shared by several integrable systems. Moreover, determinantal solutions emerge naturally from such type of equations as shown in \cite{Galleas_2016, Galleas_2016b}; and this is the feature we intend to exploit in this work for the 8VSOS model with domain-wall boundaries.

This paper is organized as follows. In \Secref{sec:CONV} we describe briefly the quantities of interest for this work and introduce conventions employed throughout this paper. We then refer the reader to \cite{Galleas_2016a, Galleas_2016b} for a more detailed discussion on those quantities. In \Secref{sec:CONV} we also present the functional equations satisfied by the partition function of interest; which have been previously derived with the aid of the dynamical Yang-Baxter algebra. \Secref{sec:FUN} is then devoted to presenting modified versions of the original functional equations and also to showing how they can be brought to the common structure above discussed. We then use \Secref{sec:DET} to discuss the resolution of such equations and concluding remarks are presented in \Secref{sec:REM}.

\section{Partition function} \label{sec:CONV}

The solid-on-solid model of interest here consists of the juxtaposition of $L\times L$ square plaquettes to which statistical weights are assigned. More precisely, we consider coordinates $(i,j) \in \mathscr{L} \times \mathscr{L}$ with $\mathscr{L} \coloneqq \{ 1, 2, \dots, L+1 \}$ and write 
$w_{i,j}$ for the statistical weight associated to the plaquette enclosed by coordinates $(i,j)$, $(i,j+1)$, $(i+1,j)$ and $(i+1,j+1)$.
Statistical fluctuations of a given plaquette $w_{i,j}$ are then characterized by the set of variables $\{ h_{i,j}, h_{i,j+1}, h_{i+1,j} , h_{i+1,j+1}     \}$ with $h_{i,j}$ usually referred to as \emph{height function}.
Moreover, here we will be considering $h_{i,j} = \tau + n_{i,j} \gamma$  with $\tau, \gamma \in \C$ and $n_{i,j} \in \Z$; and also restrict the height functions in such a way that $h_{i,j}$ and $h_{i',j'}$ at neighboring sites only differ by $\pm \gamma$.

We then want to associate a partition function to such statistical system and for that we still need to assign explicit weights $w_{i,j}$ for allowed configurations of plaquettes; and declare the boundary conditions under consideration. As for Baxter's 8VSOS model we have
\< \label{BW}
w_{ij} \begin{pmatrix} \tau \pm \gamma & \tau \\ \tau & \tau \mp \gamma \end{pmatrix} &=& [x_i - \mu_j + \gamma]  \nonumber \\
w_{ij} \begin{pmatrix} \tau \pm \gamma & \tau \pm 2\gamma \\ \tau & \tau \pm \gamma \end{pmatrix} &=& [\tau \pm \gamma] [x_i - \mu_j] [\tau]^{-1} \nonumber \\ 
w_{ij} \begin{pmatrix} \tau \pm \gamma & \tau \\ \tau & \tau \pm \gamma \end{pmatrix} &=&  [\tau \pm x_i \mp \mu_j] [\gamma] [\tau]^{-1}
\>
with $[x]$ corresponding to a Jacobi theta function. More precisely, 
\[ \label{theta}
[x] \coloneqq \frac{1}{2} \sum_{n = -\infty}^{+ \infty} (-1)^{n-\frac{1}{2}} p^{(n+\frac{1}{2})^2} e^{-(2n+1) x}
\]
with $x_i, \mu_j \in \mathbb{C}$ and fixed elliptic nome  $0 < p < 1$. According to the conventions of \cite{Whittaker_Watson_book}, the function
$[x]$ equals the Jacobi theta-function $\Theta_1 (\ii x, \nu)$ with $p = e^{\ii \pi \nu}$. Lastly, we need to declare the model's boundary conditions and here we will be considering domain-walls with $h_{1,j} = h_{j,1} = \tau + (L+1-j)\gamma$ and $h_{L+1,j} = h_{j,L+1} = \tau + (j-1)\gamma$.
In this way, the partition function of the inhomogeneous 8VSOS model with domain-wall boundaries reads
\[
\label{PF}
\mathcal{Z}_{\tau} (x_1, x_2, \dots , x_L \mid \mu_1, \mu_2, \dots , \mu_L) \coloneqq \sum_{\{ h_{i,j} \}} \prod_{i,j=1}^{L} w_{ij} \begin{pmatrix} h_{i+1,j} & h_{i+1,j+1} \\ h_{i,j} & h_{i,j+1} \end{pmatrix} \; .
\]
Although both parameters $x_i$ and $\mu_j$ correspond to lattice inhomogeneities, in the literature they are usually referred to as \emph{spectral} and \emph{inhomogeneity parameters} respectively. Moreover, $\tau$ is usually referred to as \emph{dynamical parameter} whilst $\gamma$ receives the name \emph{anisotropy parameter}. The latter is fixed in our analysis and here we omit its dependence in the LHS of \eqref{PF}.

The partition function \eqref{PF} enjoys remarkable properties due to the choice of statistical weights \eqref{BW} as discussed in \cite{Galleas_2016a, Galleas_2016b}. In particular, it is a symmetric function with respect to the variables $x_i$ and, in order to simplify our notation, it is convenient to introduce the following conventions. 
\begin{defi} \label{XXX}
Let $\mathrm{X} \coloneqq \{ x_1, x_2 , \dots , x_L \}$ and additionally write $\mathrm{X}_i^{\alpha} \coloneqq \mathrm{X} \cup \{ x_{\alpha} \} \backslash \{ x_i \}$ and $\mathrm{X}_{i,j}^{\alpha, \beta} \coloneqq \mathrm{X} \cup \{x_{\alpha}, x_{\beta} \} \backslash \{x_i, x_j \}$.
\end{defi}

\begin{rema}
Taking into account Definition \ref{XXX}, we shall then simply write $\mathcal{Z}_{\tau} ( \mathrm{X} ) = \mathcal{Z}_{\tau} (x_1, x_2, \dots , x_L \mid \mu_1, \mu_2, \dots , \mu_L)$ for the partition function \eqref{PF}.
\end{rema}

Now, given the above described conventions, in what follows we shall present the functional equations satisfied by $\mathcal{Z}_{\tau}$. The latter equations have been previously derived in \cite{Galleas_2013, Galleas_2016a, Galleas_2016c}.

\medskip

\noindent \textbf{Equation type A.} In our previous works we have used the \emph{Algebraic-Functional} framework put forward in \cite{Galleas_2010} to derive functional equations governing the partition function \eqref{PF}. One of them is referred to as \emph{equation type A} and it reads
\[
\label{eqA}
M_0^{(\mathcal{A})} \; \mathcal{Z}_{\tau - \gamma} (\mathrm{X}) + \sum_{i = 0}^L N_i^{(\mathcal{A})} \; \mathcal{Z}_{\tau} (\mathrm{X}_i^0) = 0 
\]
with coefficients given by
\< \label{coeffA}
M_0^{(\mathcal{A})} &\coloneqq& \frac{[\tau]}{[\tau + L \gamma]} \prod_{j=1}^{L} [x_0 - \mu_j] \nonumber \\
N_0^{(\mathcal{A})} &\coloneqq& -\frac{[\tau + \gamma]}{[\tau + (L+1)\gamma]} \prod_{j=1}^{L} [x_0 - \mu_j + \gamma] \prod_{j=1}^{L} \frac{[x_j - x_0 + \gamma]}{[x_j - x_0]}  \nonumber \\
N_i^{(\mathcal{A})} &\coloneqq& \frac{[\tau + \gamma + x_0 - x_i] [\tau - \gamma]}{[x_i - x_0][\tau + (L+1)\gamma]} \prod_{j=1}^{L} [x_i - \mu_j + \gamma] \prod_{\underset{j \neq i}{j=1}}^L \frac{[x_j - x_i + \gamma]}{[x_j - x_i]} \qquad i = 1,2, \dots , L \; . \nonumber \\
\>

\noindent \textbf{Equation type D.} Equation \eqref{eqA} is not the only functional equation describing  $\mathcal{Z}_{\tau}$ which can be derived along the lines of the \emph{Algebraic-Functional} method. In addition to \eqref{eqA}, here we shall also consider the \emph{equation type D} reading
\[
\label{eqD}
M_0^{(\mathcal{D})} \; \mathcal{Z}_{\tau + \gamma} (\mathrm{X}) + \sum_{i = 0}^L N_i^{(\mathcal{D})} \; \mathcal{Z}_{\tau} (\mathrm{X}_i^{0}) = 0 \; .
\]
The coefficients in \eqref{eqD} are then given by
\< \label{coeffD}
M_0^{(\mathcal{D})} &\coloneqq& \prod_{j=1}^{L} [x_{0} - \mu_j + \gamma] \nonumber \\
N_0^{(\mathcal{D})} &\coloneqq& - \prod_{j=1}^{L} [x_{0} - \mu_j] \prod_{j=1}^{L} \frac{[x_{0} - x_j + \gamma]}{[x_{0} - x_j]}  \nonumber \\
N_i^{(\mathcal{D})} &\coloneqq& \frac{[\gamma] [\tau + (L+1)\gamma + x_{0} - x_i]}{[x_{0} - x_i] [\tau + (L+1)\gamma]} \prod_{j=1}^{L} [x_i - \mu_j] \prod_{\substack{j=1 \\ j \neq i}}^L \frac{[x_i - x_j + \gamma]}{[x_i - x_j]} \qquad i = 1,2, \dots , L \; . \nonumber \\
\>

Both equations \eqref{eqA} and \eqref{eqD} are able to characterize $\mathcal{Z}_{\tau}$ in closed form as shown in \cite{Galleas_2013}. However, the dependence of such equations on the dynamical parameter $\tau$ creates some difficulties for obtaining determinant representations for our partition function directly from \eqref{eqA} and \eqref{eqD}. In the next section we shall then discuss reformulations of the equations type A and D which will help us circumventing that problem.

\section{Modified functional relations} \label{sec:FUN}

In this section we want to find modified versions of \eqref{eqA} and \eqref{eqD} where the dynamical parameter $\tau$ no longer plays the role of variable. In order to obtain such equations we shall then first proceed along the lines described in \cite{Galleas_2016b}. 

\begin{lem}[Modified equation A] \label{MeqA}
The partition function \eqref{PF} satisfies the functional relation
\[
\label{eqAneu}
\mathcal{M}_0^{(\mathcal{A})} \; \mathcal{Z}_{\tau} (\mathrm{X}) + \sum_{i=1}^L \mathcal{N}_i^{(\mathcal{A})} \; \mathcal{Z}_{\tau} (\mathrm{X}_i^0) + \sum_{i=1}^L \bar{\mathcal{N}}_i^{(\mathcal{A})} \; \mathcal{Z}_{\tau} (\mathrm{X}_i^{\bar{0}}) = 0 
\]
with coefficients defined as
\<
\label{coeffAneu} 
\mathcal{M}_0^{(\mathcal{A})} &\coloneqq&  \prod_{j=1}^L \frac{[x_{\bar{0}} - \mu_j + \gamma][x_j - x_{\bar{0}} + \gamma]}{[x_{\bar{0}} - \mu_j][x_j - x_{\bar{0}}]} - \prod_{j=1}^L \frac{[x_{0} - \mu_j + \gamma][x_j - x_{0} + \gamma]}{[x_{0} - \mu_j][x_j - x_{0}]} \nonumber \\
\mathcal{N}_i^{(\mathcal{A})} &\coloneqq&  \frac{[\gamma][\tau + \gamma + x_0 - x_i]}{[\tau + \gamma] [x_i - x_0]} \prod_{j=1}^L \frac{[x_i - \mu_j + \gamma]}{[x_0 - \mu_j]} \prod_{\substack{j = 1 \\ j \neq i}}^L \frac{[x_j - x_i + \gamma]}{[x_j - x_i]} \nonumber \\
\bar{\mathcal{N}}_i^{(\mathcal{A})} &\coloneqq&  \frac{[\gamma][\tau + \gamma + x_{\bar{0}} - x_i]}{[\tau + \gamma] [x_{\bar{0}} - x_i]} \prod_{j=1}^L \frac{[x_i - \mu_j + \gamma]}{[x_{\bar{0}} - \mu_j]} \prod_{\substack{j = 1 \\ j \neq i}}^L \frac{[x_j - x_i + \gamma]}{[x_j - x_i]} \; . \nonumber \\
\>
\end{lem}

\begin{lem}[Modified equation D] \label{MeqD}
The functional equation 
\[
\label{eqDneu}
\mathcal{M}_0^{(\mathcal{D})} \; \mathcal{Z}_{\tau} (\mathrm{X}) + \sum_{i=1}^L \mathcal{N}_i^{(\mathcal{D})} \; \mathcal{Z}_{\tau} (\mathrm{X}_i^0) + \sum_{i=1}^L \bar{\mathcal{N}}_i^{(\mathcal{D})} \; \mathcal{Z}_{\tau} (\mathrm{X}_i^{\bar{0}}) = 0 
\]
with coefficients reading
\<
\label{coeffDneu} 
\mathcal{M}_0^{(\mathcal{D})} &\coloneqq& \prod_{j=1}^L \frac{[x_{\bar{0}} - \mu_j][x_{\bar{0}} - x_j + \gamma]}{[x_{\bar{0}} - \mu_j + \gamma][x_{\bar{0}} - x_j]} - \prod_{j=1}^L \frac{[x_{0} - \mu_j][x_{0} - x_j + \gamma]}{[x_{0} - \mu_j + \gamma][x_{0} - x_j]} \nonumber \\
\mathcal{N}_i^{(\mathcal{D})} &\coloneqq&  \frac{[\gamma][\tau + (L+1)\gamma + x_0 - x_i]}{[\tau + (L+1)\gamma][x_0 - x_i]} \prod_{j=1}^L \frac{[x_i - \mu_j]}{[x_0 - \mu_j + \gamma]} \prod_{\substack{j = 1 \\ j \neq i}}^L \frac{[x_i - x_j + \gamma]}{[x_i - x_j]} \nonumber \\
\bar{\mathcal{N}}_i^{(\mathcal{D})} &\coloneqq&  \frac{[\gamma][\tau + (L+1)\gamma + x_{\bar{0}} - x_i]}{[\tau + (L+1)\gamma][x_i - x_{\bar{0}}]} \prod_{j=1}^L \frac{[x_i - \mu_j]}{[x_{\bar{0}} - \mu_j + \gamma]} \prod_{\substack{j = 1 \\ j \neq i}}^L \frac{[x_i - x_j + \gamma]}{[x_i - x_j]}   \nonumber \\
\>
is fulfilled by the partition function \eqref{PF}.
\end{lem}

\begin{proof}
The proof of Lemma \ref{MeqD} has been already presented in \cite{Galleas_2016b}, and the proof of Lemma \ref{MeqA} is analogous. 
Therefore, we shall describe here the proof of Lemma \ref{MeqA} only for completeness reasons. We start our analysis by rewriting 
\eqref{eqA} as
\[
\label{PA}
\mathcal{Z}_{\tau - \gamma} (\mathrm{X}) = - \sum_{i = 0}^L \frac{N_i^{(\mathcal{A})} (x_0, x_1, \dots , x_L)}{M_0^{(\mathcal{A})} (x_0, x_1, \dots , x_L)} \; \mathcal{Z}_{\tau} (\mathrm{X}_i^0)
\]
with $M_0^{(\mathcal{A})}  = M_0^{(\mathcal{A})} (x_0, x_1, \dots , x_L)$ and $N_i^{(\mathcal{A})}  = N_i^{(\mathcal{A})} (x_0, x_1, \dots , x_L)$ in order to emphasize the dependence of the coefficients on the relevant set of variables. Next we notice the LHS of \eqref{PA} is independent of $x_0$. Consequently, one can write
\[ \label{PAA}
\sum_{i = 0}^L \frac{N_i^{(\mathcal{A})} (x_0, x_1, \dots , x_L)}{M_0^{(\mathcal{A})} (x_0, x_1, \dots , x_L)} \; \mathcal{Z}_{\tau} (\mathrm{X}_i^0) = \sum_{i = 0}^L \frac{N_i^{(\mathcal{A})} (x_{\bar{0}}, x_1, \dots , x_L)}{M_0^{(\mathcal{A})} (x_{\bar{0}}, x_1, \dots , x_L)} \; \mathcal{Z}_{\tau} (\mathrm{X}_{\bar{i}}^{\bar{0}}) 
\]
for generic variables $x_0, x_{\bar{0}} \in \C$ and $\bar{i} = i$ for $1 \leq i \leq L$. Equation \eqref{MeqA} is then identified with \eqref{PAA}.
\end{proof}

\begin{rema}
The coefficients \eqref{coeffDneu} differ from the ones presented in \cite{Galleas_2016b} only by an overall multiplicative common factor.
\end{rema}

As far as the structure of the equations \eqref{MeqA} and \eqref{MeqD} is concerned, it is important to remark it differs significantly from the original equations \eqref{eqA} and \eqref{eqD}. For instance, in \eqref{eqA} and \eqref{eqD} the dynamical parameter $\tau$ plays the role of variable whilst it can be fixed in equations \eqref{MeqA} and \eqref{MeqD}. 
On the other hand, \eqref{MeqA} and \eqref{MeqD} also depends on an extra spectral variable, namely $x_{\bar{0}}$, in addition to the set 
$\{x_0, x_1, \dots, x_L \}$ present in the original equations \eqref{eqA} and \eqref{eqD}.
Next we remind the reader our goal here is to find an equation describing $\mathcal{Z}_{\tau}$ with the same structure as the one describing the partition function of the six-vertex model with domain-wall boundaries. Such equation for the six-vertex model can be read off from \eqref{eqA} and \eqref{eqD} simply by taking the trigonometric limit of \eqref{theta} followed  by the limit $\tau \to \infty$. At first sight it seems we would achieve such structure by setting $x_{\bar{0}} = x_0$
in \eqref{MeqA} and \eqref{MeqD}. However, the inspection of the coefficients \eqref{coeffAneu} and \eqref{coeffDneu} shows this is not the case. We shall then describe an alternative route for unveiling such equation and for that it is useful to introduce the following conventions.

\begin{defi}[Permutation] \label{pij}
Let $\mathrm{Fun}(\C^{L+2})$ be the space of meromorphic functions on $\C^{L+2}$ and $\mathfrak{S}_{L+2}$ the symmetric group acting by permutations on $(x_0 , x_1 , \dots , x_L, x_{\bar{0}}) \in \C^{L+2}$. We then write $\Pi_{i,j} \in \mathcal{S}_{L+2}$ for the $2$-cycle permuting variables $x_i$ and $x_j$.
\end{defi}

Next we proceed by inspecting the action of $\Pi_{i,j}$ on equations \eqref{MeqA} and \eqref{MeqD}. As for $1 \leq i , j \leq L$ we find our equations are invariant under the action of $\Pi_{i,j}$. On the other hand, the action of $\Pi_{\bar{0},m} \circ \Pi_{0,l}$ for $0 \leq l < m \leq L$ produces new independent equations as shown in \cite{Galleas_2016b}. More precisely, the action of $\Pi_{\bar{0},m} \circ \Pi_{0,l}$ on \eqref{MeqA} and 
\eqref{MeqD} unfold each one of those equations into the system
\<
\label{MeqP}
\mathcal{P}_0^{(l,m)} \; \mathcal{Z}_{\tau} (\mathrm{X}) + \sum_{i=1}^L \mathcal{Q}_i^{(l,m)} \; \mathcal{Z}_{\tau} (\mathrm{X}_i^0) + \sum_{i=1}^L \bar{\mathcal{Q}}_i^{(l,m)} \; \mathcal{Z}_{\tau} (\mathrm{X}_i^{\bar{0}}) + \sum_{1 \leq i, j \leq L} \mathcal{R}_{ij}^{(l,m)} \; \mathcal{Z}_{\tau} ( \mathrm{X}_{i,j}^{0, \bar{0}} ) = 0  \nonumber \\
\>
with indexes $l$ and $m$ on the interval $0 \leq l < m \leq L$. The coefficients appearing in \eqref{MeqP} require some clarifications but, before discussing them, we remark \eqref{MeqP} contains one term $\mathcal{Z}_{\tau} (\mathrm{X})$, $L$ terms of the form $\mathcal{Z}_{\tau} (\mathrm{X}_i^0)$, another $L$ terms of the form $\mathcal{Z}_{\tau} (\mathrm{X}_i^{\bar{0}})$; and $L(L-1)/2$ terms $\mathcal{Z}_{\tau} ( \mathrm{X}_{i,j}^{0, \bar{0}} )$. As for the coefficients in \eqref{MeqP}, they read
\begin{align} \label{0m}
\mathcal{P}_0^{(0, m)} & \coloneqq \Pi_{\bar{0}, m} \; \bar{\mathcal{N}}_m^{(\mathcal{T})}  & \bar{\mathcal{Q}}_j^{(0, m)} & \coloneqq \begin{cases}  \Pi_{\bar{0}, m} \; \mathcal{M}_0^{(\mathcal{T})} \qquad j=m \nonumber \\ \Pi_{\bar{0}, m} \; \bar{\mathcal{N}}_j^{(\mathcal{T})}  \qquad  \mbox{otherwise} \end{cases} \nonumber \\
\mathcal{Q}_j^{(0, m)} & \coloneqq \begin{cases}  \Pi_{\bar{0}, m} \; \mathcal{N}_m^{(\mathcal{T})}  \qquad j=m \nonumber \\ 0 \qquad \qquad \mbox{otherwise} \end{cases} & \mathcal{R}_{ij}^{(0, m)} & \coloneqq \begin{cases}  \Pi_{\bar{0}, m} \; \mathcal{N}_i^{(\mathcal{T})}  \qquad j=m \nonumber \\ \Pi_{\bar{0}, m} \; \mathcal{N}_j^{(\mathcal{T})} \qquad i=m \nonumber \\ 0 \qquad \qquad \;\;\;  \mbox{otherwise} \end{cases} \nonumber \\
\end{align}
for $1 \leq m \leq L$ and 
\begin{align} \label{lm}
\mathcal{P}_0^{(l, m)} & \coloneqq 0 & \bar{\mathcal{Q}}_j^{(l, m)} & \coloneqq \begin{cases}  \Pi_{\bar{0}, m} \circ \Pi_{0, l} \; \bar{\mathcal{N}}_l^{(\mathcal{T})}  \qquad j=l \nonumber \\ \Pi_{\bar{0}, m} \circ \Pi_{0, l} \; \mathcal{N}_l^{(\mathcal{T})}  \qquad  j=m \nonumber \\ 0 \qquad \qquad \qquad \quad \mbox{otherwise} \end{cases} \nonumber \\
\mathcal{Q}_j^{(l, m)} & \coloneqq \begin{cases}  \Pi_{\bar{0}, m} \circ \Pi_{0, l} \; \bar{\mathcal{N}}_m^{(\mathcal{T})}  \qquad j=l \nonumber \\ \Pi_{\bar{0}, m} \circ \Pi_{0, l} \; \mathcal{N}_m^{(\mathcal{T})}  \qquad j=m \nonumber \\ 0 \qquad \qquad \qquad \;\; \mbox{otherwise} \end{cases} & \mathcal{R}_{ij}^{(l, m)} & \coloneqq \begin{cases} \Pi_{\bar{0}, m} \circ \Pi_{0, l} \; \mathcal{M}_0^{(\mathcal{T})}  \qquad i=l, j=m \nonumber \\ \Pi_{\bar{0}, m} \circ \Pi_{0, l} \; \bar{\mathcal{N}}_j^{(\mathcal{T})}  \qquad \; i=l, j \neq m \nonumber \\ \Pi_{\bar{0}, m} \circ \Pi_{0, l} \; \mathcal{N}_j^{(\mathcal{T})} \qquad \; i=m \nonumber \\ \Pi_{\bar{0}, m} \circ \Pi_{0, l} \; \bar{\mathcal{N}}_i^{(\mathcal{T})} \qquad \; j=l \nonumber \\ \Pi_{\bar{0}, m} \circ \Pi_{0, l} \; \mathcal{N}_i^{(\mathcal{T})}  \qquad \; j=m, i \neq l \nonumber \\ 0 \qquad \qquad \quad \quad \quad \;\; \mbox{otherwise} \end{cases} \; . \nonumber \\
\end{align}
for $1 \leq l < m \leq L$. Also, in \eqref{0m} and \eqref{lm} we have introduced the label $\mathcal{T} \in \{ \mathcal{A} , \mathcal{D} \}$
in order to capture the implementation of our procedure for both equations \eqref{MeqA} and \eqref{MeqD}. Moreover, we remark the original equations \eqref{MeqA} and \eqref{MeqD} are also included in the structure \eqref{MeqP}; and they will be identified by the indexes $l=0$ and $m = \bar{0}$. In this way, we write
\begin{align} \label{0b0}
\mathcal{P}_0^{(0, 0)} & \coloneqq \mathcal{M}_0^{(\mathcal{T})}  & \bar{\mathcal{Q}}_j^{(0, 0)} & \coloneqq \bar{\mathcal{N}}_j^{(\mathcal{T})}  \nonumber \\
\mathcal{Q}_j^{(0, 0)} & \coloneqq \mathcal{N}_j^{(\mathcal{T})}  & \mathcal{R}_{ij}^{(0, 0)} & \coloneqq 0 \quad . \nonumber \\
\end{align}

Our next step is to solve the system \eqref{MeqP}, including the equation $(l,m) = (0, \bar{0})$, for the terms $\mathcal{Z}_{\tau} (\mathrm{X})$, 
$\mathcal{Z}_{\tau} (\mathrm{X}_i^{\bar{0}})$ ($1 \leq i \leq L$) and $\mathcal{Z}_{\tau} ( \mathrm{X}_{i,j}^{0, \bar{0}} )$ ($1 \leq i < j \leq L$). 
We clearly have enough equations for that and this procedure allows us to express each one of such functions in terms of $\mathcal{Z}_{\tau} (\mathrm{X}_i^{0})$ ($1 \leq i \leq L$). Using Cramer's method we then find neat expressions for the above-mentioned functions 
and, in particular, the solution for $\mathcal{Z}_{\tau} (\mathrm{X})$ can be written as the equation
\[ \label{HZ}
\sum_{i=0}^L \mathcal{H}_i^{(\mathcal{T})} \; \mathcal{Z}_{\tau} (\mathrm{X}_i^{0}) = 0
\]
with coefficients $\mathcal{H}_i^{(\mathcal{T})} $ expressed in terms of determinants. More precisely, $\mathcal{H}_i^{(\mathcal{T})}  = \mathrm{det}(\Omega_i^{(\mathcal{T})} )$ and in order to present the explicit form of the matrices $\Omega_i^{(\mathcal{T})}$, it is convenient to introduce the labels $m_{i,j}, n_{k,l} \colon \Z \times \Z \to \Z$ defined as
\<
m_{i,j} &\coloneqq& L\; i - \frac{i(i+1)}{2} + j +1  \nonumber \\
n_{k,l} &\coloneqq& L (k-1) - \frac{k(k+1)}{2} + l \;. 
\>
Hence, we have 
\[ \label{W0}
\Omega_0^{(\mathcal{T})}  = \begin{pmatrix} \widehat{\mathcal{P}_0} & \widehat{\mathcal{Q}} & \widehat{\mathcal{R}} \end{pmatrix}
\]
with $\widehat{\mathcal{P}_0}$ a $\frac{L(L+1)+2}{2} \times 1$ matrix with entries $(\widehat{\mathcal{P}_0})_{m_{i,j},1} \coloneqq \mathcal{P}_0^{(i, j)}$ for $i=j=0$ and $0 \leq i < j \leq L$. The matrix $\widehat{\mathcal{Q}}$, in its turn, has dimension $\frac{L(L+1)+2}{2} \times L$ and entries
$\widehat{\mathcal{Q}}_{m_{i,j},k} \coloneqq \bar{\mathcal{Q}}_k^{(i, j)}$ for $i=j=0$, $0 \leq i < j \leq L$ and $1 \leq k \leq L$.
Lastly, we have the entries $\widehat{\mathcal{R}}_{m_{i,j}, n_{k,l}} \coloneqq \mathcal{R}_{kl}^{(i, j)}$ for $i=j=0$, $0 \leq i < j \leq L$ and $1 \leq k < l \leq L$ forming the $\frac{L(L+1)+2}{2} \times \frac{L(L-1)}{2}$-dimensional matrix $\widehat{\mathcal{R}}$.
The remaining matrices $\Omega_k^{(\mathcal{T})}$ then read
\[ \label{Wk}
\Omega_k^{(\mathcal{T})} = \begin{pmatrix} \widehat{\mathcal{P}_k} & \widehat{\mathcal{Q}} & \widehat{\mathcal{R}} \end{pmatrix}
\]
with $(\widehat{\mathcal{P}_k})_{m_{i,j},1} \coloneqq \mathcal{Q}_k^{(i, j)}$ for indexes $i=j=0$ and $0 \leq i < j \leq L$. In this way, the matrix
$\widehat{\mathcal{P}_k}$ is also a $\frac{L(L+1)+2}{2} \times 1$-dimensional matrix. 

\medskip

Equation \eqref{HZ} can now be recognized as the \emph{fundamental} functional equation discussed in \Secref{sec:INTRO}. Although such equation has been shown to describe several quantities of interest, the coefficients of \eqref{HZ} still exhibit additional features that were not present in the equations previously analyzed. For instance, the coefficients $\mathcal{H}_i^{(\mathcal{T})}$ depend on the additional spectral parameter $x_{\bar{0}}$ which does not appear in the arguments of the functions $\mathcal{Z}_{\tau}$. Therefore, \eqref{HZ} corresponds to a family of fundamental functional equations with the same structure but different coefficients. The latter could then be unveiled by the series expansion of $\mathcal{H}_i^{(\mathcal{T})}$ in $x_{\bar{0}}$. Alternatively, one could also produce new fundamental equations by simply replacing $x_{\bar{0}} \mapsto z_l$ in each coefficient $\mathcal{H}_i^{(\mathcal{T})}$.
This unusual dependence on an extra parameter could also be used for solving \eqref{HZ} along the lines described in \cite{Galleas_2016b} but, before that, one would still need to clarify the number of linearly independent equations obtained through such procedures. 
In this way, here we will simply consider the approach used in \cite{Galleas_2016b} for the six-vertex model which makes use of the symmetric group action for 
enlarging \eqref{HZ} to a system of linear functional equations. The latter can then be straightforwardly used to derive determinantal solutions.

\section{Determinantal solutions} \label{sec:DET}

This section is devoted to finding the unique solution of \eqref{HZ} and the consequent determination of the partition function $\mathcal{Z}_{\tau}$. As for that we shall proceed along the lines described in \cite{Galleas_2016b} for the six-vertex model and first consider the action of $\Pi_{0,j}$ with $1 \leq j \leq \L$ on \eqref{HZ}. This procedure then uncovers the set of relations 
\[ \label{HZsys}
\sum_{i=0}^L \mathcal{H}_i^{(\mathcal{T}, \; j)} \; \mathcal{Z}_{\tau} (\mathrm{X}_i^{0}) = 0 \qquad \qquad 1 \leq j \leq L
\]
with coefficients
\<
\mathcal{H}_i^{(\mathcal{T}, \; j)} \coloneqq \begin{cases}
\Pi_{0,j} \mathcal{H}_j^{(\mathcal{T})} \qquad i=0 \\
\Pi_{0,j} \mathcal{H}_0^{(\mathcal{T})} \qquad i=j \\
\Pi_{0,j} \mathcal{H}_i^{(\mathcal{T})} \qquad i\neq 0, j
\end{cases} \; . 
\>
Next we exploit the fact that \eqref{HZsys} also exhibits the structure of a linear system of algebraic equations and use the $L$ equations contained in \eqref{HZsys} to express each function $\mathcal{Z}_{\tau} (\mathrm{X}_i^{0})$ for $1 \leq i \leq L$ in terms of $\mathcal{Z}_{\tau} (\mathrm{X})$. 
In this way, we find
\[ \label{ZtoZ}
\mathcal{Z}_{\tau} (\mathrm{X}_i^{0}) = \frac{\mathrm{det}(\mathcal{W}_i^{(\mathcal{T})})}{\mathrm{det}(\mathcal{W}_0^{(\mathcal{T})} )} \mathcal{Z}_{\tau} (\mathrm{X})
\]
where $\mathcal{W}_i^{(\mathcal{T})}$ and $\mathcal{W}_0^{(\mathcal{T})}$ are $L \times L$ matrices. Their entries are simply defined as
\<
(\mathcal{W}_0^{(\mathcal{T})})_{\alpha, \beta} &\coloneqq& \mathcal{H}_{\alpha}^{(\mathcal{T}, \; \beta)} \qquad 1 \leq \alpha , \beta \leq L \nonumber \\
(\mathcal{W}_i^{(\mathcal{T})})_{\alpha, \beta} &\coloneqq& \begin{cases}
-\mathcal{H}_{0}^{(\mathcal{T} , \; \beta)} \qquad \beta = i \\
\mathcal{H}_{\alpha}^{(\mathcal{T} , \; \beta)} \qquad \text{otherwise}
\end{cases} \; .
\>
The relation \eqref{ZtoZ} can now be regarded as an one-variable functional equation and separation of variables allows us to conclude 
$\mathcal{Z}_{\tau} (\mathrm{X}_i^{0}) = \mathrm{det}(\mathcal{W}_i^{(\mathcal{T})}) f(x_0, x_1, x_2, \dots , x_L)$ and 
$\mathcal{Z}_{\tau} (\mathrm{X}) = \mathrm{det}(\mathcal{W}_0^{(\mathcal{T})}) f(x_0, x_1, x_2, \dots , x_L)$ for a given function $f$.
As for the determination of $f$ we then use the basic conditions
\[
\frac{\partial}{\partial x_i} \left( \mathrm{det}(\mathcal{W}_i^{(\mathcal{T})}) f  \right) = 0
\]
for $0 \leq i \leq L$ as discussed in details in \cite{Galleas_2016b}. In order to present our results in a more compact manner let us also write $p_L \coloneqq L^2 (L+1)/2$ and $q_L \coloneqq (L-1)(L^2+2)/2$. In this way, by collecting the results so far we have
\< \label{repA}
\mathcal{Z}_{\tau} (\mathrm{X}_i^0) &=& - \left( \frac{[\tau + \gamma]}{[L \gamma]} \right)^{p_L} \left( \frac{[(L+1)\gamma]}{[\tau]} \right)^{q_{L}} \prod_{k=1}^{L-1} \left( \frac{[k \gamma]}{[\tau + (k+1)\gamma]} \right)^{L} \prod_{x \in \mathrm{X}_i^0} \prod_{j=1}^L [x - \mu_j + \gamma] \nonumber \\
&& \times \; \left( \frac{[\sum_{l=1}^L (x_l - \mu_l) - \gamma]}{[\sum_{l=1}^L (x_l - \mu_l) + \tau + L\gamma]}   \right) \frac{\mathrm{det}(\mathcal{W}_i^{(\mathcal{A})})  }{\left. \mathrm{det}(\mathcal{W}_i^{(\mathcal{A})})\right|_{\tau = - (L+1)\gamma}}
\>
for index $i$ in the range $0 \leq i \leq L$. Also, here we recall the entries of $\mathcal{W}_i^{(\mathcal{A})}$ are built out of the coefficients associated to the modified equation A \eqref{MeqA}.
Alternatively, one could also have considered $\mathcal{T} = \mathcal{D}$ and in that case one obtains the representations
\< \label{repD}
\mathcal{Z}_{\tau} (\mathrm{X}_i^0) &=& (-1)^{L} \frac{[\tau + (L+1)\gamma]}{[\tau + \gamma]} \left( \frac{[\tau + (L+1)\gamma]}{[L \gamma]} \right)^{p_L} \left( \frac{[(L+1)\gamma]}{[\tau + (L+2) \gamma]} \right)^{q_{L}} \prod_{k=1}^{L-1} \left( \frac{[k \gamma]}{[\tau + (k+1)\gamma]} \right)^{L} \nonumber \\
&& \times \;  \left( \frac{[\sum_{l=1}^L (x_l - \mu_l) + (L+1)\gamma]}{[\sum_{l=1}^L (x_l - \mu_l) + \tau + (L+2)\gamma]}   \right) \prod_{x \in \mathrm{X}_i^0} \prod_{j=1}^L [x - \mu_j ] \; \frac{\mathrm{det}(\mathcal{W}_i^{(\mathcal{D})})  }{\left. \mathrm{det}(\mathcal{W}_i^{(\mathcal{D})})\right|_{\tau = - \gamma}} \nonumber \\
\>
for $0 \leq i \leq L$. Here we also remark the matrices $\mathcal{W}_i^{(\mathcal{D})}$ are built out of the coefficients \eqref{coeffDneu} associated to the modified equation D \eqref{MeqD}; and explicit examples of matrices $\mathcal{W}_0^{(\mathcal{T})}$ for small lattice lengths 
can be found in \Appref{app:EG}.

\begin{rema}
Formulae \eqref{repA} and \eqref{repD} for each value of the index $i$ seem to correspond to an independent representation since
$\Pi_{0,i} \; \mathrm{det}(\mathcal{W}_i^{(\mathcal{T})})$ does not seem to be trivially related to $\mathrm{det}(\mathcal{W}_0^{(\mathcal{T})})$. In this way, each one of the formulae \eqref{repA} and \eqref{repD} describes $L+1$ two-parameters 
families of continuous representations.
\end{rema}

\section{Concluding remarks} \label{sec:REM}

In this work we have elaborated on the \emph{Algebraic-Functional} approach to the partition function of the 8VSOS model with domain-wall boundary conditions. In particular, we have shown the latter partition function obeys the same type of functional equation as its six-vertex model counterpart. The significance of this result is twofold: on the one hand we have enlarged the diversity of quantities governed by such type of linear functional equation; and on the other hand, we have consequently obtained new determinantal representations for the aforementioned partition function. 

Although the functional equation in discussion, namely \eqref{HZ}, exhibits the same structure as the equations previously studied in the context of the six-vertex model, it still possesses additional features that were not encountered before. For instance, its coefficients 
$\mathcal{H}_i^{(\mathcal{T})}$ also depend non-trivially on an extra spectral parameter playing no role in the arguments of the functions
$\mathcal{Z}_{\tau}$. In this way, we can regard \eqref{HZ} as a continuous family of functional equations as discussed in \Secref{sec:FUN}.
Moreover, one advantage of having some quantity described by equations with the structure of \eqref{HZ} is that determinantal solutions follow naturally with little dependence on the explicit form of the equation's coefficients. In this way, we have used \eqref{HZ} to find novel 
determinantal representations for the partition function of the 8VSOS model with domain-wall boundaries. 

Our new determinants also exhibit some unusual features that have not been found previously in the literature to the best of our knowledge. For instance, here we have obtained determinants of $L\times L$ matrices whose entries are also determinants. These sub-determinants are in their turn taken over $[\frac{L(L+1)+2}{2}] \times [\frac{L(L+1)+2}{2}]$ matrices and at first look this feature suggests Dodgson's condensation method plays some role in our results. However, the existence of a precise relation with the condensation method still remains unclear to us.

\bibliographystyle{alpha}
\bibliography{references}

\begin{thebibliography}{LBT13b}

\bibitem[Bax71]{Baxter_1971}
R.~J. Baxter.
\newblock Eight vertex model in lattice statistics.
\newblock {\em Phys. Rev. Lett.}, 26:832, 1971.

\bibitem[Bax07]{Baxter_book}
R.~J. Baxter.
\newblock {\em {Exactly Solved Models in Statistical Mechanics}}.
\newblock Dover Publications, Inc., Mineola, New York, 2007.

\bibitem[DIK13]{Deift_2013}
P.~Deift, A.~Its, and I.~Krasovsky.
\newblock {Toeplitz Matrices and Toeplitz Determinants under the Impetus of the
  Ising Model: Some History and Some Recent Results}.
\newblock {\em {Comm. Pure Appl. Math.}}, 66(9):1360--1438, 2013.

\bibitem[Gal10]{Galleas_2010}
W.~Galleas.
\newblock Functional relations for the six-vertex model with domain wall
  boundary conditions.
\newblock {\em J. Stat. Mech.}, 06:P06008, 2010.

\bibitem[Gal12]{Galleas_2012}
W.~Galleas.
\newblock {Multiple integral representation for the trigonometric SOS model
  with domain wall boundaries}.
\newblock {\em {Nucl. Phys. B}}, {858}({1}):{117--141}, {2012}.

\bibitem[Gal13]{Galleas_2013}
W.~Galleas.
\newblock {Refined functional relations for the elliptic SOS model}.
\newblock {\em {Nucl. Phys. B}}, {867}:{855--871}, {2013}.

\bibitem[Gal14]{Galleas_SCP}
W.~Galleas.
\newblock {Scalar product of Bethe vectors from functional equations}.
\newblock {\em {Comm. Math. Phys.}}, {329}({1}):{141--167}, {2014}.

\bibitem[Gal15]{Galleas_openSCP}
W.~Galleas.
\newblock {Off-shell scalar products for the $XXZ$ spin chain with open
  boundaries}.
\newblock {\em Nucl. Phys. B}, 893:346--375, {2015}.

\bibitem[Gal16a]{Galleas_2016}
W.~Galleas.
\newblock {New differential equations in the six-vertex model}.
\newblock {\em {J. Stat. Mech.}}, ({3}):33106--33118, {2016}.

\bibitem[Gal16b]{Galleas_2016b}
W.~Galleas.
\newblock {On the elliptic $\mathfrak{gl}_2 $ solid-on-solid model: functional
  relations and determinants}.
\newblock {\em to appear in J. Math. Phys.}, arXiv: 1606.06144 [math-ph], 2016.

\bibitem[Gal16c]{Galleas_2016a}
W.~Galleas.
\newblock Partition function of the elliptic solid-on-solid model as a single
  determinant.
\newblock {\em Phys. Rev. E}, 94(1):010102, 2016.

\bibitem[Gal17]{Galleas_2016c}
W.~Galleas.
\newblock {Continuous representations of scalar products of Bethe vectors}.
\newblock {\em J. Math. Phys.}, 58(8):083504, 2017.

\bibitem[Gal18]{Galleas_2017}
W.~Galleas.
\newblock {Six-vertex model and non-linear differential equations I. Spectral
  problem}.
\newblock {\em {Comm. Math. Phys.}}, {363}:{59--96}, {2018}.

\bibitem[GL14]{Galleas_Lamers_2014}
W.~Galleas and J.~Lamers.
\newblock {Reflection algebra and functional equations}.
\newblock {\em {Nucl. Phys. B}}, {886}({0}):{1003--1028}, {2014}.

\bibitem[Ize87]{Izergin_1987}
A.~G. Izergin.
\newblock Partition function of the six-vertex model in a finite lattice.
\newblock {\em Sov. Phys. Dokl.}, 32:878, 1987.

\bibitem[KBI93]{Korepin_book}
V.~E. Korepin, N.~M. Bogoliubov, and A.~G. Izergin.
\newblock {\em Quantum inverse scattering method and correlation functions}.
\newblock Cambridge University Press, 1993.

\bibitem[KO49]{Onsager_Kauf_1949}
B.~Kaufman and L.~Onsager.
\newblock {Crystal statistics III. Short-range order in a binary Ising
  lattice}.
\newblock {\em {Phys. Rev.}}, {76}({8}):{1244--1252}, {1949}.

\bibitem[Kor82]{Korepin_1982}
V.~E. Korepin.
\newblock {Calculation of norms of Bethe wave functions}.
\newblock {\em Commun. Math. Phys.}, 86:391--418, 1982.

\bibitem[Lam15]{Lamers_2015}
J.~Lamers.
\newblock {Integral formula for elliptic SOS models with domain walls and a
  reflecting end}.
\newblock {\em {Nucl. Phys. B}}, {901}:{556--583}, {2015}.

\bibitem[Lax68]{Lax_1968}
P.~D. Lax.
\newblock {Integrals of nonlinear equations of evolution and solitary waves}.
\newblock {\em Comm. Pure Applied Math.}, 21:467--490, 1968.

\bibitem[LBT13a]{Terras_2013a}
D.~Levy-Bencheton and V.~Terras.
\newblock {An algebraic Bethe ansatz approach to form factors and correlation
  functions of the cyclic eight-vertex solid-on-solid model}.
\newblock {\em J. Stat. Mech.}, 04:P04015, 2013.

\bibitem[LBT13b]{Terras_2013b}
D.~Levy-Bencheton and V.~Terras.
\newblock {Spontaneous staggered polarizations of the cyclic solid-on-solid
  model from the algebraic Bethe Ansatz}.
\newblock {\em {J. Stat. Mech.}}, {10}:{P10012}, {2013}.

\bibitem[Lie67]{Lieb_1967}
E.~H. Lieb.
\newblock {Residual entropy of square lattice}.
\newblock {\em Phys. Rev.}, 162(1):162, 1967.

\bibitem[LW72]{Lieb_review}
E.~H. Lieb and F.~Y. Wu.
\newblock {Two-dimensional ferroelectric models}.
\newblock In C.~Domb and M.~Green, editors, {\em {Phase Transitions and
  Critical Phenomena}}, volume~{1}, pages {331--490}. {Academic Press}, {1972}.

\bibitem[PRS08]{Pakuliak_2008}
S.~Pakuliak, V.~Rubtsov, and A.~Silantyev.
\newblock {SOS model partition function and the elliptic weight function}.
\newblock {\em J. Phys. A}, 41:295204, 2008.

\bibitem[Ros09]{Rosengren_2009}
H.~Rosengren.
\newblock {An Izergin-Korepin type identity for the 8VSOS model with
  applications to alternating sign matrices}.
\newblock {\em Adv. Appl. Math.}, 43:137--155, 2009.

\bibitem[WW27]{Whittaker_Watson_book}
E.~T. Whittaker and G.~N Watson.
\newblock {\em A Course of Modern Analysis}.
\newblock Cambridge University Press, fourth edition, 1927.

\end{thebibliography}

\appendix
%
%
%
%

\section{Matrices $\mathcal{W}_0^{(\mathcal{T})}$} \label{app:EG}
 
In this appendix we give explicit expressions for the matrices $\mathcal{W}_0^{(\mathcal{T})}$ entering formulae \eqref{repA} and \eqref{repD} for small lattice lenghts. As for the case $L=1$, the matrix $\mathcal{W}_0^{(\mathcal{A})}$ is a $1 \times 1$ matrix with single entry reading
\<
\left| \begin{matrix}
\frac{[x_0 - x_{\bar{0}} + \gamma][x_{\bar{0}} - \mu_1 + \gamma]}{[x_0 - x_{\bar{0}}][x_{\bar{0}} - \mu_1]}  - \frac{[x_0 - x_{1} + \gamma][x_{1} - \mu_1 + \gamma]}{[x_0 - x_{1}][x_{1} - \mu_1]} &  \frac{[\gamma] [x_{\bar{0}} - x_0 + \tau +  \gamma] [x_0 - \mu_1 + \gamma]} {[\tau + \gamma] [x_{\bar{0}} - x_0] [x_{\bar{0}} - \mu_1]} \\
\frac{[\gamma] [x_0 - x_{\bar{0}}  + \tau +  \gamma] [x_{\bar{0}} - \mu_1 + \gamma]} {[\tau + \gamma] [x_0 - x_{\bar{0}}] [x_0 - \mu_1]}  & 
\frac{[x_{\bar{0}} - x_0  + \gamma][x_0 - \mu_1 + \gamma]}{[x_{\bar{0}} - x_0][x_0 - \mu_1]}  - \frac{[x_{\bar{0}} - x_{1} + \gamma][x_{1} - \mu_1 + \gamma]}{[x_{\bar{0}} - x_{1}][x_{1} - \mu_1]}
\end{matrix} \right| \; . \nonumber \\
\>
 Similarly, the single entry of $\mathcal{W}_0^{(\mathcal{D})}$ for $L=1$ is given by
 \<
\left| \begin{matrix}
\frac{[x_{\bar{0}} - x_0  + \gamma][x_{\bar{0}} - \mu_1]}{[x_{\bar{0}} - x_0][x_{\bar{0}} - \mu_1 + \gamma]}  - \frac{[x_1 - x_{0} + \gamma][x_{1} - \mu_1]}{[x_1 - x_{0}][x_{1} - \mu_1 + \gamma]} &  \frac{[\gamma] [x_{\bar{0}} - x_0 + \tau +  2\gamma] [x_0 - \mu_1]} {[\tau + 2\gamma] [x_0 - x_{\bar{0}}] [x_{\bar{0}} - \mu_1 + \gamma]} \\
\frac{[\gamma] [x_0 - x_{\bar{0}}  + \tau +  2\gamma] [x_{\bar{0}} - \mu_1]} {[\tau + 2\gamma] [x_{\bar{0}} - x_0] [x_0 - \mu_1 + \gamma]}  & 
\frac{[x_0 - x_{\bar{0}} + \gamma][x_0 - \mu_1]}{[x_0 - x_{\bar{0}}][x_0 - \mu_1 + \gamma]}  - \frac{[x_1 - x_{\bar{0}} + \gamma][x_{1} - \mu_1]}{[x_1 - x_{\bar{0}}][x_{1} - \mu_1 + \gamma]}
\end{matrix} \right| \; . \nonumber \\
\>
 
Now turning our attention to the case $L=2$, the matrices $\mathcal{W}_0^{(\mathcal{A})}$ and $\mathcal{W}_0^{(\mathcal{D})}$ are $2 \times 2$ matrices with entries given by $4 \times 4$ determinants. In order to conveniently present such matrices we then write 
 \[
 \mathcal{W}_0^{(\mathcal{T})} = \begin{pmatrix}
 | \mathcal{K}_1^{(\mathcal{T})} | & | \mathcal{K}_2^{(\mathcal{T})} | \\
 | \mathcal{K}_3^{(\mathcal{T})} | & | \mathcal{K}_4^{(\mathcal{T})} |  
 \end{pmatrix}
 \]
 and introduce functions
 \<
 \mathcal{U}^{i,j}_{k,l} &\coloneqq& \frac{[\gamma] [x_i - x_k + \gamma] [x_j - x_k + \tau + \gamma]}{[\tau + \gamma][x_i - x_k] [x_j - x_k]} \prod_{t=1}^2 \frac{[x_k - \mu_t + \gamma]}{[x_l - \mu_t]} \nonumber \\
  \bar{\mathcal{U}}^{i,j}_{k,l} &\coloneqq& -\frac{[\gamma] [x_k - x_i + \gamma] [x_j - x_k + \tau + 3\gamma]}{[\tau + 3\gamma][x_k - x_i] [x_j - x_k]} \prod_{t=1}^2 \frac{[x_k - \mu_t]}{[x_l - \mu_t + \gamma]} \nonumber \\
  \mathcal{V}^{i,j}_{k,l} &\coloneqq& \prod_{s \in \{k,l\}} \frac{[x_s - x_i + \gamma]}{[x_s - x_i]} \prod_{t=1}^2 \frac{[x_i - \mu_t + \gamma]}{[x_i - \mu_t]} -  \prod_{s \in \{k,l\}} \frac{[x_s - x_j + \gamma]}{[x_s - x_j]} \prod_{t=1}^2 \frac{[x_j - \mu_t + \gamma]}{[x_j - \mu_t]} \nonumber \\
  \bar{\mathcal{V}}^{i,j}_{k,l} &\coloneqq& \prod_{s \in \{k,l\}} \frac{[x_i - x_s + \gamma]}{[x_i - x_s]} \prod_{t=1}^2 \frac{[x_i - \mu_t]}{[x_i - \mu_t + \gamma]} -  \prod_{s \in \{k,l\}} \frac{[x_j - x_s + \gamma]}{[x_j - x_s]} \prod_{t=1}^2 \frac{[x_j - \mu_t]}{[x_j - \mu_t + \gamma]} \; .\nonumber \\
 \>
 In this way, we have
 \begin{align}
 \mathcal{K}_1^{(\mathcal{A})} = & \begin{pmatrix}
 \mathcal{V}^{\bar{0}, 1}_{0,2} & \mathcal{U}^{2, \bar{0}}_{0, \bar{0}} & \mathcal{U}^{0, \bar{0}}_{2, \bar{0}} & 0 \\
 \mathcal{U}^{2,0}_{\bar{0}, 0} & \mathcal{V}^{0,1}_{\bar{0},2} & \mathcal{U}^{\bar{0},0}_{2,0} & - \mathcal{U}^{\bar{0}, 1}_{2,1} \\
 \mathcal{U}^{0, 2}_{\bar{0}, 2} & \mathcal{U}^{\bar{0},2}_{0,2} & \mathcal{V}^{2,1}_{0, \bar{0}} & - \mathcal{U}^{\bar{0}, 1}_{0,1} \\
 0 & \mathcal{U}^{\bar{0}, 2}_{1, 2} & - \mathcal{U}^{\bar{0}, 0}_{1, 0} & \mathcal{V}^{2,0}_{\bar{0},1} 
 \end{pmatrix} &
 \mathcal{K}_2^{(\mathcal{A})} &= \begin{pmatrix}
 -\mathcal{U}^{0, 1}_{2,1} & \mathcal{U}^{2, \bar{0}}_{0, \bar{0}} & \mathcal{U}^{0, \bar{0}}_{2, \bar{0}} & 0 \\
 0 & \mathcal{V}^{0,1}_{\bar{0},2} & \mathcal{U}^{\bar{0},0}_{2,0} & - \mathcal{U}^{\bar{0}, 1}_{2,1} \\
 -\mathcal{U}^{0, 1}_{\bar{0}, 1} & \mathcal{U}^{\bar{0},2}_{0,2} & \mathcal{V}^{2,1}_{0, \bar{0}} & - \mathcal{U}^{\bar{0}, 1}_{0,1} \\
 -\mathcal{U}^{1, 0}_{\bar{0}, 0} & \mathcal{U}^{\bar{0}, 2}_{1, 2} & - \mathcal{U}^{\bar{0}, 0}_{1, 0} & \mathcal{V}^{2,0}_{\bar{0},1} 
 \end{pmatrix} \nonumber \\ \nonumber \\
 \mathcal{K}_3^{(\mathcal{A})} = &\begin{pmatrix}
- \mathcal{U}^{0, 2}_{1,2} & \mathcal{U}^{0, \bar{0}}_{1, \bar{0}} & \mathcal{U}^{1, \bar{0}}_{0, \bar{0}} & 0 \\
- \mathcal{U}^{0,2}_{\bar{0}, 2} & \mathcal{V}^{1,2}_{0,\bar{0}} & \mathcal{U}^{\bar{0},1}_{0,1} & - \mathcal{U}^{\bar{0}, 2}_{0,2} \\
 0 & \mathcal{U}^{\bar{0},0}_{1,0} & \mathcal{V}^{0,2}_{\bar{0}, 1} & - \mathcal{U}^{\bar{0}, 2}_{1,2} \\
 \mathcal{U}^{2, 0}_{\bar{0}, 0}  & \mathcal{U}^{\bar{0}, 0}_{2, 0} & - \mathcal{U}^{\bar{0}, 1}_{2, 1} & \mathcal{V}^{0,1}_{\bar{0},2} 
 \end{pmatrix} &
 \mathcal{K}_4^{(\mathcal{A})} & = \begin{pmatrix}
 \mathcal{V}^{\bar{0}, 2}_{0,1} & \mathcal{U}^{0, \bar{0}}_{1, \bar{0}} & \mathcal{U}^{1, \bar{0}}_{0, \bar{0}} & 0 \\
 \mathcal{U}^{0,1}_{\bar{0}, 1} & \mathcal{V}^{1,2}_{0,\bar{0}} & \mathcal{U}^{\bar{0},1}_{0,1} & - \mathcal{U}^{\bar{0}, 2}_{0,2} \\
 \mathcal{U}^{1, 0}_{\bar{0}, 0} & \mathcal{U}^{\bar{0},0}_{1,0} & \mathcal{V}^{0,2}_{\bar{0},1} & - \mathcal{U}^{\bar{0}, 2}_{1,2} \\
 0 & \mathcal{U}^{\bar{0}, 0}_{2, 0} & - \mathcal{U}^{\bar{0}, 1}_{2, 1} & \mathcal{V}^{0,1}_{\bar{0},2}
 \end{pmatrix} \; . \nonumber \\
\end{align}
The matrices $\mathcal{K}_i^{(\mathcal{D})}$ in their turn are given by
\begin{align}
 \mathcal{K}_1^{(\mathcal{D})} = & \begin{pmatrix}
 \bar{\mathcal{V}}^{\bar{0}, 1}_{0,2} & \bar{\mathcal{U}}^{2, \bar{0}}_{0, \bar{0}} & \bar{\mathcal{U}}^{0, \bar{0}}_{2, \bar{0}} & 0 \\
 \bar{\mathcal{U}}^{2,0}_{\bar{0}, 0} & \bar{\mathcal{V}}^{0,1}_{\bar{0},2} & \bar{\mathcal{U}}^{\bar{0},0}_{2,0} & - \bar{\mathcal{U}}^{\bar{0}, 1}_{2,1} \\
 \bar{\mathcal{U}}^{0, 2}_{\bar{0}, 2} & \bar{\mathcal{U}}^{\bar{0},2}_{0,2} & \bar{\mathcal{V}}^{2,1}_{0, \bar{0}} & - \bar{\mathcal{U}}^{\bar{0}, 1}_{0,1} \\
 0 & \bar{\mathcal{U}}^{\bar{0}, 2}_{1, 2} & - \bar{\mathcal{U}}^{\bar{0}, 0}_{1, 0} & \bar{\mathcal{V}}^{2,0}_{\bar{0},1} 
 \end{pmatrix} &
 \mathcal{K}_2^{(\mathcal{D})} &= \begin{pmatrix}
 -\bar{\mathcal{U}}^{0, 1}_{2,1} & \bar{\mathcal{U}}^{2, \bar{0}}_{0, \bar{0}} & \bar{\mathcal{U}}^{0, \bar{0}}_{2, \bar{0}} & 0 \\
 0 & \bar{\mathcal{V}}^{0,1}_{\bar{0},2} & \bar{\mathcal{U}}^{\bar{0},0}_{2,0} & - \bar{\mathcal{U}}^{\bar{0}, 1}_{2,1} \\
 -\bar{\mathcal{U}}^{0, 1}_{\bar{0}, 1} & \bar{\mathcal{U}}^{\bar{0},2}_{0,2} & \bar{\mathcal{V}}^{2,1}_{0, \bar{0}} & - \bar{\mathcal{U}}^{\bar{0}, 1}_{0,1} \\
 -\bar{\mathcal{U}}^{1, 0}_{\bar{0}, 0} & \bar{\mathcal{U}}^{\bar{0}, 2}_{1, 2} & - \bar{\mathcal{U}}^{\bar{0}, 0}_{1, 0} & \bar{\mathcal{V}}^{2,0}_{\bar{0},1} 
 \end{pmatrix} \nonumber \\ \nonumber \\
 \mathcal{K}_3^{(\mathcal{D})} = &\begin{pmatrix}
- \bar{\mathcal{U}}^{0, 2}_{1,2} & \bar{\mathcal{U}}^{0, \bar{0}}_{1, \bar{0}} & \bar{\mathcal{U}}^{1, \bar{0}}_{0, \bar{0}} & 0 \\
- \bar{\mathcal{U}}^{0,2}_{\bar{0}, 2} & \bar{\mathcal{V}}^{1,2}_{0,\bar{0}} & \bar{\mathcal{U}}^{\bar{0},1}_{0,1} & - \bar{\mathcal{U}}^{\bar{0}, 2}_{0,2} \\
 0 & \bar{\mathcal{U}}^{\bar{0},0}_{1,0} & \bar{\mathcal{V}}^{0,2}_{\bar{0}, 1} & - \bar{\mathcal{U}}^{\bar{0}, 2}_{1,2} \\
 \bar{\mathcal{U}}^{2, 0}_{\bar{0}, 0}  & \bar{\mathcal{U}}^{\bar{0}, 0}_{2, 0} & - \bar{\mathcal{U}}^{\bar{0}, 1}_{2, 1} & \bar{\mathcal{V}}^{0,1}_{\bar{0},2} 
 \end{pmatrix} &
 \mathcal{K}_4^{(\mathcal{D})} & = \begin{pmatrix}
 \bar{\mathcal{V}}^{\bar{0}, 2}_{0,1} & \bar{\mathcal{U}}^{0, \bar{0}}_{1, \bar{0}} & \bar{\mathcal{U}}^{1, \bar{0}}_{0, \bar{0}} & 0 \\
 \bar{\mathcal{U}}^{0,1}_{\bar{0}, 1} & \bar{\mathcal{V}}^{1,2}_{0,\bar{0}} & \bar{\mathcal{U}}^{\bar{0},1}_{0,1} & - \bar{\mathcal{U}}^{\bar{0}, 2}_{0,2} \\
 \bar{\mathcal{U}}^{1, 0}_{\bar{0}, 0} & \bar{\mathcal{U}}^{\bar{0},0}_{1,0} & \bar{\mathcal{V}}^{0,2}_{\bar{0},1} & - \bar{\mathcal{U}}^{\bar{0}, 2}_{1,2} \\
 0 & \bar{\mathcal{U}}^{\bar{0}, 0}_{2, 0} & - \bar{\mathcal{U}}^{\bar{0}, 1}_{2, 1} & \bar{\mathcal{V}}^{0,1}_{\bar{0},2}
 \end{pmatrix} \; . \nonumber \\
\end{align} 
 
 \vskip 3cm

\end{document}